\documentclass[copyright,creativecommons]{eptcs}
\providecommand{\event}{GANDALF 2011} 
\usepackage{breakurl}             

\usepackage{amsmath}
\usepackage{amsthm}
\usepackage{amssymb}
\usepackage{xspace}

\newcommand{\dettagli}[1]{}

\newtheorem{remark}{Remark}

\newtheorem{proposition}{Proposition}
\newtheorem{theorem}{Theorem}
\newtheorem{lemma}{Lemma}
\newtheorem{corollary}{Corollary}

\makeatletter
\def\twoheadrightarrowfill@{\arrowfill@\relbar\relbar\twoheadrightarrow}
\newcommand{\xtwoheadrightarrow}[2][]{\ext@arrow 0359\twoheadrightarrowfill@{#1}{#2}}
\makeatother

\newcommand{\der}[1]{\ensuremath{\xrightarrow{#1}}}

\newcommand \tpl[1]{\langle #1 \rangle}

\def\AVPA{\text{\sffamily\upshape AVPA}\xspace}

\def\state{\text{\sffamily state}\xspace}
\def\TOP{\text{\sffamily top}\xspace}
\def\vis{\text{\sffamily vis}\xspace}
\def\children{\text{\sffamily children}\xspace}

\def\PMC{\text{{\sffamily PMC}}\xspace}
\def\CTL{\text{{\sffamily CTL}}\xspace}
\def\LTL{\text{{\sffamily LTL}}\xspace}
\def\CTLSTAR{\text{{\sffamily CTL$^{*}$}}\xspace}
\def\ECTL{\text{{\sffamily ECTL}}\xspace}
\def\ACTL{\text{{\sffamily ACTL}}\xspace}
\def\OPD{\text{{\sffamily OPD}}\xspace}
\def\OPDS{\text{{\sffamily OPDs}}\xspace}
\def\CON{\text{{\sffamily CON}}\xspace}
\def\ID{\text{{\sffamily ID}}\xspace}
\def\length{\text{{\sffamily length}}\xspace}
\def\word{\text{{\sffamily word}}\xspace}
\def\suffix{\text{{\sffamily suffix}}\xspace}

\newcommand\System{{\cal S}}
\newcommand\N{{\mathbb{N}}}

\newcommand{\X}{\mathsf{X}}     
\newcommand{\E}{\mathsf{E}}     
\newcommand{\Aop}{\mathsf{A}}     
\newcommand{\uU}{\mathsf{U}}    
\newcommand{\U}{{\,\uU\,}}      
\newcommand{\F}{\mathsf{F}}     
\newcommand{\G}{\mathsf{G}}     

\def\EXPTIME{{\sc Exptime}}

\def\TWOEXPTIME{{\sc 2Exptime}}
\def\EXPSPACE{{\sc Expspace}}

\title{New results on pushdown module checking\\ with imperfect information}
\author{Laura Bozzelli
\institute{Technical University of Madrid (UPM),
28660 Boadilla del Monte, Madrid, SPAIN}
}
\def\titlerunning{New results on pushdown module checking with imperfect information}
\def\authorrunning{Laura Bozzelli}
\begin{document}
\maketitle

\begin{abstract}
Model checking of open pushdown systems (\OPD) w.r.t.~standard branching temporal logics (\emph{pushdown module checking} or \PMC)
has been recently investigated in the literature, both in the context of environments with perfect and imperfect information
about the system (in the last case, the environment has only a partial view of the system's control states and stack content). For standard $\CTL$, \PMC with imperfect information is known to be undecidable. If the stack content is assumed to be \emph{visible}, then
the problem is decidable and \TWOEXPTIME-complete (matching the complexity of \PMC with perfect information against $\CTL$).
The decidability status of \PMC with imperfect information against $\CTL$ restricted to the case where the depth of the stack content is visible is  open. In this paper, we show that with this restriction, \PMC with imperfect information against $\CTL$ remains undecidable. On the other hand, we individuate an interesting subclass of \OPDS with visible stack content depth such that  \PMC with imperfect information  against the  existential fragment of $\CTL$ is decidable
and in \TWOEXPTIME. Moreover, we show that the \emph{program complexity}  of \PMC  with imperfect information and visible stack content against $\CTL$ is
\TWOEXPTIME-complete (hence, exponentially harder than the program complexity of \PMC with perfect information, which is known to be \EXPTIME-complete).

\end{abstract}

\section{Introduction}

\noindent \textbf{Verification of open systems.} In the literature, formal verification  of open systems is in general formulated as two-players games (between the system and the environment). This setting is suitable when the correctness requirements on the behavior of the system are formalized by linear-time temporal logics. In order to take into account also requirements expressible in branching-time temporal logics, recently,
Kupferman, Vardi, and Wolper \cite{KV96,KVW01} introduce the \emph{module checking} framework for the verification of finite-state open systems.
In such a framework, the open finite-state system is described by a labeled state-transition graph called \emph{module}, whose set of states is partitioned into a set of \emph{system states} (where the system makes a transition) and a set of \emph{environment states} (where the environment makes a transition). Given a module $\mathcal{M}$ describing the system to be verified, and a branching-time temporal formula $\varphi$ specifying the desired behavior of the system, the \emph{module checking problem} asks whether for all possible environments, $\mathcal{M}$ satisfies $\varphi$. In particular, it might be that the environment does not enable all the external nondeterministic choices. Module checking thus involves not only checking that the full computation tree $T_{\mathcal{M}}$ obtained by unwinding $\mathcal{M}$ (which corresponds to the interaction of $\mathcal{M}$ with a maximal environment) satisfies the specification $\varphi$, but also that every tree obtained from it by pruning children of environment nodes (this corresponds to disable possible environment choices) satisfy $\varphi$. In \cite{KV97} module checking for finite-state systems has been extended to a setting where the environment has \emph{imperfect information} about the states of the system (see also \cite{Rei84,CH05} for related work regarding imperfect information). In this setting, every state of the module is a composition of \emph{visible} and \emph{invisible} variables where the latter are hidden to the environment. Thus, the composition of a module $\mathcal{M}$ with an environment with imperfect information corresponds to a tree obtained from $T_{\mathcal{M}}$ by pruning children of environment nodes in such a way that the pruning is consistent with the partial information available to the environment. One of the results in \cite{KV97} is that $\CTL$ finite-state module checking with imperfect information has the same complexity as $\CTL$ finite-state module checking with perfect information, i.e., it is \EXPTIME-complete, but its \emph{program complexity} (i.e., the complexity of the problem in terms of
the size of the system) is exponentially harder, i.e. \EXPTIME-complete. \vspace{0.2cm}

\noindent \textbf{Pushdown module checking.} An active field of research is model­-checking of  pushdown
systems. These represent an infinite-state formalism suitable to
model the control flow of recursive sequential programs.
The model checking problem of (closed) pushdown systems against standard regular temporal logics (such as $\LTL$, $\CTL$,  $\CTLSTAR$, or the modal  $\mu$-calculus)
 is decidable and it has been intensively studied in
recent years leading to efficient verification algorithms and
tools (see for example~\cite{Wal96,BEM97,BR00}). Recently, in~\cite{BMP10,AMV07,FMP07},  the module checking framework has been extended to the class of \emph{open pushdown systems} (\OPD), i.e. pushdown systems in which the set of configurations is partitioned (in accordance with the control state and the symbol on the top of the stack) into a set of \emph{system configurations} and a set of \emph{environment configurations}. \emph{Pushdown module checking} (\PMC, for short) against standard branching temporal logics, like $\CTL$ and $\CTLSTAR$, has been investigated both in the context of environments with perfect information~\cite{BMP10} and imperfect information~\cite{AMV07,FMP07}
about the system (in the last case, the environment has only a partial view of the system's control states and stack content). For the perfect information setting, as in the case of finite-state systems, \PMC is much harder than standard pushdown model checking for both $\CTL$ and $\CTLSTAR$. For example, for $\CTL$, while  pushdown model checking is \EXPTIME-complete \cite{Wal00}, \PMC with perfect information is \TWOEXPTIME-complete~\cite{BMP10} (however, the program complexities of the two problems are the same, i.e., \EXPTIME-complete~\cite{Boz07a,BMP10}).
For the imperfect information setting, \PMC against $\CTL$  is in general undecidable~\cite{AMV07}, and undecidability relies on hiding information about the stack content. 
The decidability status for the last problem restricted to the class of \OPDS where the \emph{stack content depth} is visible is left open in~\cite{AMV07}.
On the other hand, \PMC with imperfect information against $\CTL$ restricted to the class of \OPDS with imperfect information about the internal control states, but a visible stack content, is decidable  and has the same complexity as \PMC with perfect information. However, its program complexity  is open: it lies somewhere between \EXPTIME\ and \TWOEXPTIME~\cite{AMV07}.\vspace{0.2cm}

\noindent \textbf{Our contribution.} We establish new results on \PMC with imperfect information against $\CTL$. Moreover, we also consider a subclass of \OPDS, we call \emph{stable} \OPDS, 
where the transition relation is consistent with the partial information  available to the environment. Our main results are the following.
\begin{itemize}
  \item The \emph{program complexity} of \PMC with imperfect information against $\CTL$ restricted to the class of \OPDS with \emph{visible stack content} is
  \TWOEXPTIME-hard,\footnote{hence,  \TWOEXPTIME-complete, since \PMC with imperfect information against $\CTL$ restricted to the class of \OPDS with \emph{visible stack content} is known to be \TWOEXPTIME-complete~\cite{AMV07}} even for a fixed formula of the existential fragment $\ECTL$ of $\CTL$ (hence, exponentially harder than the program complexity of \PMC with perfect information against $\CTL$, which is known to be \EXPTIME-complete~\cite{BMP10}). The result is obtained by a polynomial-time reduction from the acceptance problem for \EXPSPACE-bounded Alternating Turing Machines, which is known to be \TWOEXPTIME-complete \cite{CKS81}.
  \item \PMC with imperfect information against $\CTL$ restricted to the class of \OPDS with \emph{visible stack content depth} is undecidable, even if the $\CTL$ formula is assumed to be in the fragment  of $\CTL$ using only temporal modalities $\E\F$ and $\E\X$, and their duals, and the \OPD is assumed to be \emph{stable} and having only environment configurations. The result is obtained by a reduction from  the Post's Correspondence Problem, a well known undecidable problem \cite{HU79}.
  \item \PMC with imperfect information against the \emph{existential fragment} $\ECTL$ of $\CTL$ restricted to the class of \emph{stable} \OPDS with \emph{visible stack content depth} and having only environment configurations is instead decidable and in \TWOEXPTIME.
The result is proved by a reduction to non-emptiness of B\"{u}chi alternating visible pushdown automata (\AVPA) \cite{Boz07}, which is  \TWOEXPTIME-complete \cite{Boz07}.
\end{itemize}

The full version of this paper can be asked to the author by e-mail.

\section{Preliminaries}

Let $\N$ be the set of natural numbers. A tree $T$ is a prefix closed subset of $\N^{*}$.
The elements of $T$ are called \emph{nodes} and the empty word $\varepsilon$ is the \emph{root} of $T$. For $x\in T$, the set of \emph{children} of $x$ (in $T$) is $\children(T,x)=\{x\cdot i\in T\mid i\in \N\}$. For $x\in T$, a (full) \emph{path} of $T$ from $x$ is a maximal sequence $\pi=x_1,x_2,\ldots$ of nodes in $T$ such that $x_1=x$ and for each $1\leq i<|\pi|$, $x_{i+1}\in\children(T,x_i)$ .
 In the following, for a path of $T$, we mean a path of $T$ from the root $\varepsilon$. For an alphabet $\Sigma$, a $\Sigma$-labeled tree is a pair $\tpl{T,V}$, where $T$ is a tree and
$V:T\rightarrow \Sigma$ maps each node of $T$ to a symbol in $\Sigma$. Given two $\Sigma$-labeled trees $\tpl{T,V}$ and $\tpl{T',V'}$, we say that
$\tpl{T,V}$ \emph{is contained in} $\tpl{T',V'}$ if $T\subseteq T'$ and $V'(x)=V(x)$ for each $x\in T$.   In order to simplify the notation, sometimes we write simply $T$ to denote a $\Sigma$-labeled tree $\tpl{T,V}$.

\subsection{Module checking with imperfect information}

In this paper we consider \emph{open systems}, i.e. systems that interact
with their environment and whose behavior depends on this
interaction.
Moreover, we consider the case where the environment has imperfect information about the states of the system.
This is modeled by an equivalence relation $\cong$ on the set of states. States that are indistinguishable by the environment, because the difference between them is kept invisible by the system, are equivalent according to $\cong$. We describe an open system by an
 \emph{open} Kripke structure (called also \emph{module}
\cite{KVW01}) $\mathcal{M} = \tpl{AP,S=S_{sy}\cup S_{en},  s_0, R, L,\cong}$,
where $AP$ is a finite set of atomic propositions, $S$ is a (possibly infinite)  set of  states partitioned into a set $S_{sy}$
of \emph{system} states and a set $S_{en}$ of \emph{environment}
states, and $s_0\in S$ is a designated
initial state. Moreover, $R\subseteq
S\times S$ is a  transition relation,  $L:S\rightarrow 2^{AP}$ maps each state $s$
to the set of atomic propositions that hold in $s$, and $\cong$ is an equivalence relation on the set of states $S$. Since the designation of a state as an environment state is obviously known to the environment, we require that for all states $s,s'$ such that $s\cong s'$, $s\in S_{en}$ iff $s'\in S_{en}$. For each $s\in S$, we denote by $\vis(s)$ the equivalence class of $s$ w.r.t.~$\cong$. Intuitively, $\vis(s)$ represents what the environment ``sees" of $s$. A successor of $s$ is a state $s'$ such that $(s,s')\in R$. State $s$ is \emph{terminal} if it has no successor.
 When the module $\mathcal{M}$ is in
a non-terminal \emph{system} state $s\in S_{sy}$, then all the successors of $s$ are possible next states. On the other hand, when $\mathcal{M}$
is in a non-terminal \emph{environment} state $s\in S_{en}$, then the environment decides, based on the visible part of each successor of $s$, and of the history of the computation so far, to which of the successor states the computation can proceed, and to which it can not.
Additionally, we consider environments that cannot block the system, i.e. not all the transitions from a non-terminal environment state are disabled.
For a state $s$ of $\mathcal{M}$, let $T_{\mathcal{M},s}$ be the \emph{computation tree of $\mathcal{M}$ from $s$}, i.e. the $S$-labeled tree
obtained by unwinding $\mathcal{M}$ starting from $s$ in the usual way. Note that $T_{\mathcal{M},s}$ describes the behavior of $\mathcal{M}$  under the \emph{maximal} environment, i.e. the environment that never restricts the set of next states. The behavior of $\mathcal{M}$  under a specific environment (possibly different from the maximal one) is formalized by the notion of \emph{strategy tree} as follows.
For a node $x$ of the computation tree $T_{\mathcal{M},s}$, let $s_1,\ldots,s_p$ be the sequence of states labeling the partial path from the root to node $x$. We denote by $\vis(x)$ the sequence $\vis(s_1),\ldots,\vis(s_p)$, which represents the visible part of the (partial) computation $s_1,\ldots,s_p$ associated with node $x$.
A \emph{strategy tree from $s$} is a $S$-labeled tree obtained from the computation tree
$T_{\mathcal{M},s}$
by pruning from
$T_{\mathcal{M},s}$ subtrees whose roots are children  of
nodes labeled by  environment states. Additionally, we require that
such a pruning is consistent with the partial information available to the environment: if two nodes $x_1$ and $x_2$ of $T_{\mathcal{M},s}$ are indistinguishable, i.e. $\vis(x_1)=\vis(x_2)$, then the subtree rooted at $x_1$ is pruned iff the subtree rooted at $x_2$ is pruned as well.
Formally, a strategy tree of $\mathcal{M}$ from a state $s\in S$ is a $S$-labeled tree $ST$ such that $ST$ is contained in $T_{\mathcal{M},s}$ and the following holds:
\begin{itemize}
  \item for each node $x$ of $ST$ labeled by a \emph{system} state, $\children(ST,x)=\children(T_{\mathcal{M},s},x)$;
  \item for each node $x$ of $ST$ labeled by an \emph{environment} state, \mbox{$\children(ST,x)\neq \emptyset$ if $\children(T_{\mathcal{M},s},x)\neq\emptyset$;}
  \item for all nodes $x_1$ and $x_2$ of $T_{\mathcal{M},s}$ such that $\vis(x_1)=\vis(x_2)$, $x_1$ is a node of $ST$ iff $x_2$ is a node of $ST$.
  Note that if $x_1$ is a child of an environment node, then so is $x_2$.
\end{itemize}

\noindent For a   node $x$ of $ST$, $\state(x)$ denotes the $S$-state labeling $x$. A strategy tree of $\mathcal{M}$ is a strategy tree of $\mathcal{M}$ from the initial state.   In the following, a strategy tree $ST$ is seen as a $2^{AP}$-labeled tree, i.e. taking the
label of a node $x$ to be $L(\state(x))$. We also consider a restricted class of modules. A module $\mathcal{M}$ is \emph{stable} (\emph{w.r.t.~visible information}) iff for all states $s_1$ and $s_2$ s.t. $\vis(s_1)=\vis(s_2)$ and both $s_1$ and $s_2$ have some successor, it holds that:
for each successor $s'_1$ of $s_1$, there is a successor $s'_2$ of $s_2$ s.t. $\vis(s'_1)=\vis(s'_2)$. Note that this notion is similar to that given in \cite{Rei84} for standard imperfect information games.\vspace{0.1cm}


\noindent \textbf{$\CTL$ Module Checking:} as specification logical language, we consider
the standard branching temporal logic $\CTL$~\cite{CE81}, whose formulas $\varphi$ over   $AP$  are assumed to be in positive normal form, i.e. defined as: \vspace{0.1cm}

\text{\hspace{0.2cm}}$\varphi:= \mathtt{true}\, | \, prop \, | \, \neg prop \, | \, \varphi \vee\varphi \, | \, \varphi \wedge\varphi \, | \, \E\X\varphi\, | \, \Aop\X\varphi \, | \, \E(\varphi \U\varphi) \, | \, \Aop(\varphi \U\varphi)\, | \, \E(\varphi \widetilde{\U}\varphi) \, | \, \Aop(\varphi \widetilde{\U}\varphi)$\vspace{0.1cm}

\noindent where $prop\in AP$, $\E$ (resp., $\Aop$) is the existential (resp., universal) path quantifier, 
$\X$ and $\U$ are the next and until temporal operators, and $\widetilde{\U}$ is the dual of 
$\U$.
We  use classical shortcuts: $\E\F\varphi$ is for $\E(\mathtt{true}\U \varphi)$ (``existential eventually") and
$\Aop\F\varphi$ is for $\Aop(\mathtt{true}\U \varphi)$ (``universal eventually"), and their duals
$\Aop\G\varphi:=\neg\E\F\neg\varphi$ 
and $\E\G\varphi:=\neg\Aop\F\neg\varphi$. 
We also consider  the universal (resp., existential) fragment $\ACTL$ (resp., $\ECTL$) of $\CTL$ obtained by disallowing the existential (resp., universal) path quantifier, and the fragment $\CTL(\E\F,\E\X,\Aop\G,\Aop\X)$ using only temporal modalities $\E\F$ and $\E\X$, and their duals.
For a definition of the  semantics of $\CTL$ (which is given with respect to $2^{AP}$-labeled trees) see~\cite{CE81}.

For a module $\mathcal{M}$ and a $\CTL$ formula $\varphi$ over $AP$,  $\mathcal{M}$ \emph{reactively satisfies} $\varphi$, denoted $\mathcal{M}\models_r\varphi$, if all the strategy trees of $\mathcal{M}$ (from the initial state) satisfy $\varphi$.   
Note that $\mathcal{M} \not\models_r \varphi$ is
\emph{not} equivalent to $\mathcal{M} \models_r \neg\varphi$. Indeed, $\mathcal{M}
\not\models_r \varphi$ just states that there is some strategy tree $ST$  satisfying $\neg\varphi$.

\subsection{Pushdown Module Checking with Imperfect Information}

In this paper we consider Modules induced by Open Pushdown Systems (\OPD, for short),
i.e., Pushdown systems where the set of configurations is partitioned (in accordance with the
control state and the symbol on the top of the stack) into a set of environment configurations
and a set of system configurations.

An \OPD is a tuple $\System =\tpl{AP,	Q,q_0,\Gamma,\flat,\Delta,\mu,Env}$, where $AP$ is a finite set of propositions,
$Q$	is a finite set of control states, $q_0\in Q$ is the initial control state, $\Gamma$ is a finite stack alphabet,
 $\flat\notin \Gamma$ is the \emph{special stack
bottom symbol}, $\Delta\subseteq (Q\times Q)\cup (Q\times Q\times \Gamma)\cup (Q\times (\Gamma\cup \{\flat\})\times Q)$ is the transition relation,
 $\mu:Q\times (\Gamma\cup\{\flat\})\rightarrow 2^{AP}$ is a labeling function, and
 $Env\subseteq Q\times (\Gamma\cup\{\flat\})$ is used to  specify the set of environment configurations.
  A transition of the form $(q,q',\gamma)$, written $q\der{\mathtt{push}(\gamma)}q'$, is a push transition, where $\gamma\neq \flat$ is pushed onto the stack  (and the control changes from $q$ to $q'$). A transition of the form $(q,\gamma,q')$, written $q\der{\mathtt{pop}(\gamma)}q'$, is a pop transition, where $\gamma$ is popped from the stack. 
   Finally, a transition of the form $(q,q')$, written $q\der{}q'$, is an \emph{internal} transition, where the stack is not used.
  We assume that
 $Q\subseteq 2^{I\cup H}$, where $I$ and $H$ are disjoint finite sets of \emph{visible} and \emph{invisible} \emph{control variables},
 and $\Gamma \subseteq 2^{I_\Gamma\cup H_\Gamma}$, where $I_\Gamma$ and $H_\Gamma$ are disjoint finite sets of \emph{visible} and \emph{invisible} \emph{stack content variables}. 

 A \emph{configuration or state} of $\System$ is a pair $(q,\alpha)$, where
$q \in Q$  and $\alpha\in \Gamma^{*}\cdot\flat$ is a stack content.
We denote by $\TOP(\alpha)$ the \emph{top of the stack content $\alpha$}, i.e. the leftmost symbol of $\alpha$.
For a control  state $q\in Q$, \emph{the visible part of $q$} is $\vis(q)=q\cap I$. For a stack symbol $\gamma\in \Gamma$, if $\gamma\subseteq H_\Gamma$
and $\gamma\neq \emptyset$, we set $\vis(\gamma)=\varepsilon$, otherwise we set $\vis(\gamma)=\gamma\cap I_\Gamma$. By setting $\vis(\gamma)=\varepsilon$ whenever $\gamma$ consists entirely of invisible variables, we allow the system to completely hide a push operation. 
The \emph{visible part of a configuration $(q,\alpha)$} is $(\vis(q),\vis(\alpha))$, where  for $\alpha=\gamma_0\ldots \gamma_n\cdot\flat$,
$\vis(\alpha)=\vis(\gamma_0)\ldots \vis(\gamma_n)\cdot\flat$. The \emph{stack content} (resp., the \emph{control}) \emph{is visible} if $H_\Gamma=\emptyset$ (resp.,
$H=\emptyset$). Moreover, the \emph{stack content depth is visible} if $\vis(\gamma)\neq \varepsilon$ for each stack symbol $\gamma\in \Gamma$.
Since  the designation of an \OPD state as an environment state is known to the environment, we require that for all states $(q,\alpha)$ and $(q',\alpha')$ such that $(\vis(q),\vis(\TOP(\alpha)))=(\vis(q'),\vis(\TOP(\alpha')))$,  $(q,\TOP(\alpha))\in Env$ iff
$(q',\TOP(\alpha'))\in Env$.
The \OPD $\System$ induces an infinite-state module
\mbox{$\mathcal{M}_\System = \tpl{AP,S=S_{sy}\cup S_{en},  s_0, R, L,\cong}$,} defined as follows:
\begin{itemize}
  \item $S_{sy}\cup S_{en}$ is the set of configurations of $\System$, and $S_{en}$ is the set of states $(q,\alpha)$ s.t. $(q,\TOP(\alpha))\in Env$;
  \item $s_0=(q_0,\flat)$ is the initial configuration (initially, the stack is empty);
  \item $((q,\alpha),(q',\alpha'))\in R$ iff: or (1) $q\der{}q'\in\Delta$ and $\alpha'=\alpha$, or (2) \mbox{$q\der{\mathtt{push}(\gamma)}q'\in\Delta$} and $\alpha'=\gamma\cdot\alpha$, or (3)
$q\der{\mathtt{pop}(\gamma)}q'\in\Delta$, and either $\alpha'=\alpha=\gamma=\flat$ or $\gamma\neq\flat$ and $\alpha=\gamma\cdot\alpha'$
  (note that every pop transition that removes $\flat$ also pushes it back);
  \item $L((q,\alpha))=\mu((q,\TOP(\alpha)))$ for all $(q,\alpha)\in S$;
 \item for all  $(q,\alpha),(q',\alpha')\in S$, we have that $(q,\alpha)\cong (q',\alpha')$ iff $(\vis(q),\vis(\alpha))=(\vis(q'),\vis(\alpha'))$.
\end{itemize}

\noindent A strategy tree of $\mathcal{S}$ is a strategy tree of $\mathcal{M}_{\mathcal{S}}$ from the initial state. Given  $(q,\gamma)\in Q\times (\Gamma\cup \{\flat\})$,  $(q,\gamma)$ is \emph{non-terminal} (w.r.t.~$\mathcal{S}$) iff: or $q\der{}q'\in\Delta$ or $q\der{\mathtt{pop}(\gamma)}q'\in\Delta$ or
$q\der{\mathtt{push}(\gamma')}q'\in\Delta$ for some $q'\in Q$ and $\gamma'\in\Gamma$.
Note that a state $(q,\alpha)$ of $\mathcal{S}$ has some successor (in $\mathcal{M}_\System$) iff   $(p,\TOP(\alpha))$ is non-terminal.
We also consider a subclass of \OPD. An \OPD $\mathcal{S}=\tpl{AP,	Q,q_0,\Gamma,\flat,\Delta,\mu,Env}$ is \emph{stable} iff for all non-terminal pairs
$(q_1,\gamma_1),(q_2,\gamma_2)\in Q\times (\Gamma\cup \{\flat\})$ s.t. $\vis(q_1)=\vis(q_2)$ and $\vis(\gamma_1)=\vis(\gamma_2)$, the following holds:
\begin{itemize}
  \item if $q_1\der{}q'_1\in \Delta$, then there is $q_2\der{}q'_2\in \Delta$ such that $\vis(q'_1)=\vis(q'_2)$;
  \item if $q_1\der{\mathtt{push}(\gamma)}q'_1\in \Delta$, then there is $q_2\der{\mathtt{push}(\gamma')}q'_2\in \Delta$ such that $\vis(q'_1)=\vis(q'_2)$ and $\vis(\gamma)=\vis(\gamma')$;
  \item if $q_1\der{\mathtt{pop}(\gamma_1)}q'_1\in \Delta$, then there is $q_2\der{\mathtt{pop}(\gamma_2)}q'_2\in \Delta$ such that $\vis(q'_1)=\vis(q'_2)$.
\end{itemize}

\begin{remark}\label{remark:stable} Note that for a \OPD $\mathcal{S}$ with visible stack content depth, $\mathcal{S}$ is stable iff  $\mathcal{M}_\System$ is stable.
\end{remark}

In the rest of this paper, we consider \OPD $\mathcal{S}$ where each state is labeled by a singleton in $2^{AP}$ (for a given set $AP$ of atomic propositions), hence, the strategy trees can be seen as $AP$-labeled trees.

The \emph{pushdown module checking problem $($\PMC$)$ with imperfect information against $\CTL$}  is to decide, for a given
\OPD $\System$ and a $\CTL$ formula $\varphi$, whether $\mathcal{M}_\System\models_r\varphi$.

\section{Pushdown module checking  for $\OPD$ with visible stack content}\label{sec:Problem1}

In this section, we prove the following result.
\begin{theorem}\label{theorem:programComplexity} The \emph{program complexity} of $\PMC$ with imperfect information against $\CTL$  restricted to the class of \OPDS with \emph{visible stack content} is  \TWOEXPTIME-hard, even for a fixed $\ECTL$ formula.\footnote{for program complexity, we mean the complexity of the problem in terms of the size of the $\OPD$, for a fixed $\CTL$ formula}
\end{theorem}

Theorem~\ref{theorem:programComplexity}  is proved by a
polynomial-time reduction from the acceptance problem for \EXPSPACE-bounded \emph{alternating Turing Machines} (TM) with a binary branching degree, which is known to be \TWOEXPTIME-complete \cite{CKS81}.
In the rest of this section, we fix  such a TM machine $\mathcal{T}=\tpl{A,Q=Q_{\forall}\cup Q_{\exists},q_0,\delta,F}$, where $A$ is the input alphabet containing  the blank symbol $\#$,
$Q_\exists$ (resp., $Q_{\forall}$) is the set of existential (resp., universal) states,   $q_0$ is the initial state,   $\delta: Q\times A\rightarrow (Q \times A \times
\{\leftarrow,\rightarrow\})\times(Q \times A \times \{\leftarrow,\rightarrow\})$ is the transition function, and $F\subseteq Q$ is the set  of
accepting states. Thus, in each step, $\mathcal{T}$ overwrites the tape cell being scanned, and the tape head moves one position to the left ($\leftarrow$) or right ($\rightarrow$).
We fix an input $w_{in}\in A^{*}$ and consider the  parameter $n=|w_{in}|$ (we assume that $n>1$).
Since $\mathcal{T}$ is \EXPSPACE-bounded, we can assume that $\mathcal{T}$ uses exactly  $2^{n}$ tape cells when started on the input $w_{in}$. Hence,
a  TM configuration (of $\mathcal{T}$ over $w_{in}$)  is  a word $C=w_1\cdot (a,q)\cdot w_2\in A^*\cdot(A\times Q)\cdot
A^*$ of length exactly $2^{n}$ denoting that the tape content is $w_1\cdot a\cdot w_2$, the current state is $q$, and the tape head is at position $|w_1|+1$.  $C$ is \emph{accepting} if the associated state $q$ is in $F$. We denote by $succ_L(C)$ (resp., $succ_R(C)$) the TM successor of $C$ obtained by choosing the left (resp., right) triple in $\delta(q,a)$. The initial configuration $C_{in}$ is $(w_{in}(0),q_0),w_{in}(1),\ldots,w_{in}(n-1),\#,\#,\ldots,\#$, where the number of blanks at the right of $w_{in}(n-1)$ is  $2^{n}-n$ .
For a TM configuration $C=C(0),\ldots,C(2^{n}-1)$,  the
         `value' $u_i$ of the $i$-{th} symbol of $succ_L(C)$ (resp., $succ_R(C)$) is completely determined by the values
         $C(i-1)$, $C(i)$ and $C(i+1)$ (taking $C(i+1)$ for $i=2^{n}-1$ and
         $C(i-1)$ for $i=0$ to be some special symbol, say $\bot$). We denote by
         $next_L(C(i-1),C(i),C(i+1))$ (resp.,
         $next_R(C(i-1),C(i),C(i+1))$) our
         expectation for $u_i$ (these functions can be trivially obtained from the transition function  $\delta$ of
         $\mathcal{T}$).

We prove the following result, hence, Theorem~\ref{theorem:programComplexity} follows (note that $\ECTL$ is the dual of $\ACTL$).

\begin{theorem}\label{theorem:mainLowerBound} One can construct in polynomial time $($in the sizes of $\mathcal{T}$ and  $w_{in}$$)$ an \OPD $\mathcal{S}$ with \emph{visible stack content}  such that \emph{$\mathcal{T}$ accepts $w_{in}$} iff there is a strategy tree
of  $\mathcal{S}$ satisfying a \emph{fixed} computable $\ACTL$ formula $\varphi$ $($independent on $\mathcal{T}$ and  $w_{in}$$)$. 
\end{theorem}

In the following, first we describe a suitable encoding of acceptance of $\mathcal{T}$ over $w_{in}$. Then, we illustrate the construction of the \OPD of Theorem~\ref{theorem:mainLowerBound} based on this encoding.\vspace{0.15cm}

\noindent \textbf{Preliminary step: encoding of acceptance of $\mathcal{T}$ over $w_{in}$}. We use the following set $\Gamma$ of symbols (which will correspond to the stack alphabet of the \OPD $\mathcal{S}$ of Theorem~\ref{theorem:mainLowerBound}):\footnote{Since the stack content of $\mathcal{S}$ is visible, we assume that each stack symbol in $\Gamma$  consists exactly of a visible stack content variable. Hence, we identify the set $\Gamma$ of stack symbols with the set of visible stack content variables.}\vspace{0.1cm}

\noindent \text{\hspace{4.3cm}}$\Gamma = \Lambda\cup \{L,R,0,1,\exists,\forall\}\cup (\{\natural\}\times \{\bot,1,\ldots,n\})$\vspace{0.1cm}

\noindent where $\Lambda$ consists of the triples $(u_p,u,u_s)$ such that $u\in A\cup (A\times Q)$ and $u_p,u_s\in A\cup (A\times Q)\cup \{\bot\}$. Intuitively, $u_p,u,u_s$ represent three consecutive symbols in a TM configuration $C$, where $u_p=\bot$ (resp., $u_s=\bot$) iff $u$ is the first (resp., the last) symbol of $C$.  First, we describe the encoding of TM configurations $C=C(0),\ldots,C(2^{n}-1)$ by finite words over $\Gamma$. Intuitively, the encoding of $C$ is a sequence of $2^{n}$ blocks, where  the $i$-th block ($0\leq i\leq 2^{n}-1$) keeps tracks of the triple $(C(i-1),C(i),C(i+1))$ and the binary code of position $i$ (cell number). Note that the cell numbers are in
the range $[0,2^n-1]$ and can be encoded by using $n$ bits. Formally, a \emph{TM block} is a word over $\Gamma$ of length $n+2$ of the form $bl= t,bit_1,\ldots,bit_n,(\natural,l_\bot)$, where $t\in\Lambda$, $bit_1,\ldots,bit_n\in \{0,1\}$, and $l_\bot$ is the position $i$ of the first bit $bit_i$ (from left to right) such that $bit_i=0$ if such a 0-bit exists, and $l_\bot=\bot$ otherwise. The \emph{content} $\CON(bl)$ of $bl$ is $t$ and the \emph{block number} $\ID(bl)$ of $bl$ is the integer in $[0,2^n-1]$ whose binary code is $bit_1,\ldots,bit_n$ (we assume that the first bit is the least significant one). Fix a \emph{pseudo} TM configuration $C=C(0),\ldots,C(k-1)$ with $k>1$, which is defined as a TM configuration with the unique difference that the length $k$ of $C$ is not required to be $2^{n}$. We say that $C$ is \emph{initial} if $C$ corresponds to the initial TM configuration $C_{in}$ with the unique difference that the number of blanks at the right of $w_{in}(n-1)$ is not required to be $2^{n}-n$. A TM \emph{pseudo code} of $C$ is a word  $w_C=bl_0\cdot\ldots\cdot bl_{k-1}\cdot tag$ over $\Gamma$  satisfying the following, where $C(-1),C(k)=\bot$:
\begin{itemize}
  \item $tag\in \{\exists,\forall\}$ and $tag=\exists$ iff $C$ is \emph{existential} (i.e., the associated TM state is in $Q_\exists$);
  \item each $bl_i$ is a TM block such that $\CON(bl_i)=(C(i-1),C(i),C(i+1))$;
  \item $\ID(bl_0)=0$ and $\ID(bl_{k-1})=2^{n}-1$. Moreover, for each $0\leq h<k-1$, $\ID(bl_{h})\neq 2^{n}-1$.
\end{itemize}
  If $k=2^{n}$ and additionally, for each $i$, $ID(bl_i)=i$, then we say that the word  $w_C$ is the TM \emph{code} of the TM configuration $C$.
Given a non-empty sequence $\nu=C_1,\ldots,C_p$ of pseudo TM configurations, a \emph{pseudo sequence-code of $\nu$} is a word over $\Gamma\cup \{\flat\}$ (recall that $\flat$ is the special bottom stack symbol of an \OPD) of the form $w_\nu=\flat\cdot w_{C_1}\cdot dir_2\cdot w_{C_2}\cdot\ldots\cdot dir_p\cdot w_{C_p}$ such that $dir_2,\ldots,dir_p\in \{L,R\}$ and each $w_{C_i}$ is a pseudo code of $C_i$. The word $w_\nu$ is \emph{initial} if $C_1$  is initial, and is \emph{accepting} if $C_p$ is accepting and each $C_j$ with $j<p$ is not accepting.  Moreover, if, additionally, each $C_i$ is a TM configuration and $w_{C_i}$ is a code of $C_i$, then we say that $w_\nu$ is a \emph{sequence-code}. Furthermore, $w_\nu$ is \emph{faithful to the evolution of $\mathcal{T}$} if  $C_i=succ_{dir_i}(C_{i-1})$ for each $2\leq i\leq p$.
 We encode the acceptance of $\mathcal{T}$ over $w_{in}$ as follows, where a  $\Gamma\cup \{\flat\}$-labeled tree is
\emph{minimal} if the children of each node have distinct labels. An \emph{accepting pseudo tree-code} 
is a \emph{finite} minimal $\Gamma\cup \{\flat\}$-labeled tree $T$ such that for each path $\pi$ of $T$, the word labeling $\pi$, written $w_\pi$, is an initial and accepting pseudo sequence-code (of some sequence of pseudo TM configurations) and:
\begin{itemize}
  \item each internal node labeled by $\exists$ (\emph{existential choice node}) has at most two children: one, if any, is labeled by $L$, and the other one, if any, is labeled by $R$;
  \item each internal node labeled by $\forall$ (\emph{universal choice node}) has exactly two children: one is labeled by $L$, and the other one is labeled by $R$.
\end{itemize}

If for each path $\pi$ of $T$, $w_\pi$ is a \emph{sequence-code}, then we say that $T$ is an \emph{accepting tree-code}. Moreover, if for each path $\pi$ of $T$, $w_\pi$ is faithful to the evolution of $\mathcal{T}$, then we say that $T$ is \emph{fair}.
\begin{remark} \label{lemma:acceptance}
$\mathcal{T}$ accepts $w_{in}$ iff there is an accepting fair tree-code.
\end{remark}

\noindent \textbf{Construction of the \OPD $\mathcal{S}$ of Theorem~\ref{theorem:mainLowerBound}.}
We construct the \OPD $\mathcal{S}$ in a modular way, i.e. $\mathcal{S}$ is obtained by putting together three \OPD $\mathcal{S}_0,\mathcal{S}_1$, and $\mathcal{S}_2$. Intuitively, the first \OPD $\mathcal{S}_0$ does not use invisible information and ensures that the set of its \emph{finite} strategy trees is precisely the set of accepting pseudo tree-codes. The second \OPD $\mathcal{S}_1$, which  does not use invisible information, is used to check, together with a fixed $\ACTL$ formula, that an accepting pseudo tree-code is in fact an accepting tree-code. The last \OPD $\mathcal{S}_2$, which is the unique `component' which uses invisible information, is used to check, together with a fixed $\ACTL$ formula, that an accepting tree-code is fair. First, we consider the \OPDS $\mathcal{S}_0$ and $\mathcal{S}_1$.  For a finite word $w$, we denote by $w^{R}$ the reverse of $w$.

\begin{lemma}\label{lemma:STEp1}
One can build in polynomial time $($in the sizes of $\mathcal{T}$ and $w_{in}$$)$ an \OPD $\mathcal{S}_0$  with \emph{no invisible information}, stack alphabet $\Gamma$,  set of propositions $\Gamma\cup \{\flat\}$, and special \emph{terminal}\footnote{a terminal control state is a control state from which there is no transition} control state $p_{fin}$ s.t.
$\mathcal{S}_0$ has only push transitions and  the set of its \emph{finite} strategy trees $ST$    is the set of accepting pseudo tree-codes. Moreover, for each node $x$ of $ST$, the stack content of $\state(x)$ is the \emph{reverse} of the word labeling the partial path from the root to $x$, and  $\state(x)$ has control state $p_{fin}$ and it is a system state if $x$ is a leaf.
 \end{lemma}


\begin{lemma}\label{lemma:STEp2}
One can build in polynomial time $($in the sizes of $\mathcal{T}$ and $w_{in}$$)$ an \OPD $\mathcal{S}_1$   with \emph{no invisible information}, stack alphabet $\Gamma$, and set of propositions $\{main_1,check_1,good_1\}$ s.t.~$\mathcal{S}_1$ has only pop transitions and for each state $s=(p_0,\alpha^{R})$ such that  $p_0$ is the initial control state and $\alpha$ is a TM \emph{pseudo sequence-code}, the following holds:
 $s$ is labeled by $main_1$, there is a unique strategy tree $ST$ 
 from $s$, $ST$ is \emph{finite}, and $\alpha$ is a \emph{sequence code} iff $ST$ satisfies the fixed $\ACTL$ formula $\varphi_{check_1}=\Aop\G(check_1\rightarrow \Aop\F good_1)$.
\end{lemma}


\begin{lemma}\label{lemma:STEp3}
One can build in polynomial time $($in the sizes of $\mathcal{T}$ and $w_{in}$$)$ an  \OPD $\mathcal{S}_2$ with \emph{invisible information} and \emph{visible stack content},   stack alphabet $\Gamma$, and set of  propositions $AP=\{main_2,check_2,select_2,$ $good_2\}$,  s.t.
$\mathcal{S}_2$ has only pop transitions and for each state $s=(p_0,\alpha^{R})$, where $p_0$ is the initial control state and $\alpha$ is a TM \emph{sequence-code}, the following holds:
state $s$ is labeled by $main_2$,
each strategy tree  of $\mathcal{S}_2$ from $s$ is finite, and $\alpha$ is \emph{faithful to the evolution of $\mathcal{T}$} iff there is a strategy tree $ST$ from $s$ satisfying the fixed $\ACTL$ formula $\varphi_{check_2}=\Aop\G\bigl(check_2 \rightarrow [((\Aop\X\,check_2)\vee (\Aop\X \,select_2)) \wedge \Aop\F\,good_2] \bigr)$.
\end{lemma}
 \begin{proof}
 We informally describe the construction of  $\mathcal{S}_2$, 
 which additionally satisfies the following: (1) the labeling function can be seen as a mapping $\mu:P\rightarrow AP$, where $P$ is the set of control states, and (2) for each control state $p$, $\vis(p)=\mu(p)$. Assume that initially $\mathcal{S}_2$ is in state $(p_0,\alpha^{R})$, where $p_0$ is the initial control state and $\alpha$ is a sequence-code. Note that $\alpha$ is faithful to the evolution of $\mathcal{T}$ iff for each subword\footnote{given a  word $w$, a finite word $w'$ is a \emph{subword} of $w$ if $w$ can be written in the form $w=w_1\cdot w'\cdot w_2$} of $\alpha^{R}$ of the form
 $(bl_1^{R}\cdot\beta^{R}_1)\cdot dir\cdot \beta^{R}_2$ such that $\beta_1\cdot bl_1$ is a prefix of a TM code, $bl_1$ is a TM block with $\CON(bl_1)=(u_{1,p},u_1,u_{1,s})$, and $\beta_2$ is a TM code, the following holds: $u_1=next_{dir}(u_{2,p},u_2,u_{2,s})$, where
$(u_{2,p},u_2,u_{2,s})=\CON(bl_2)$ and $bl_2$ is the unique TM block of $\beta_2$ such that  $\ID(bl_2)=\ID(bl_1)$.
Then, starting from the $main_2$-state $(p_0,\alpha^{R})$, the $main_2$-copy of $\mathcal{S}_2$ pops $\alpha^{R}$ (symbol by symbol) and terminates its computation (a $main_2$-state is labeled by $main_2$) with the additional ability to start by \emph{internal nondeterminism} (i.e., the choices are made by the system) $n$ auxiliary copies (each of them in a $check_2$-state) whenever the popped symbol is in $\{\natural\}\times \{\bot,1,\ldots,n\}$. Let $l_\bot^{1}$ be the currently popped symbol in
$\{\natural\}\times \{\bot,1,\ldots,n\}$. Hence, the current stack content is of the form $bl_1^{R}\cdot \alpha'$, where $bl_1$ is a TM block.
Assume that $\alpha'$ contains some symbol in $\{L,R\}$ (the other case being simpler), hence $\alpha'$ is of the form
$\beta^{R}_1\cdot dir\cdot \beta^{R}_2\cdot \alpha''$ such that $\beta_1\cdot bl_1$ is a prefix of a TM code, $bl_1$ is a TM block with $\CON(bl_1)=(u_{1,p},u_1,u_{1,s})$, and $\beta_2$ is a TM code. Then,
the $i$-th \emph{$check_2$} copy ($1\leq i\leq n$), which visits states labeled by $check_2$, deterministically pops the stack (symbol by symbol) until the symbol $dir$  and memorizes by its finite control the $i$-th bit $bit^{1}_i$ of $bl_1$ and the symbol $u_1$ in the content $\CON(bl_1)$ of $bl_1$.
When the symbol $dir\in \{L,R\}$ is popped, then the $i$-th $check_2$ copy pops $\beta_2^{R}$ 
 and terminates its computation with the additional ability to  start by \emph{external} nondeterminism
  (i.e., the choices are made by the environment) an auxiliary copy of $\mathcal{S}_2$ in a $select_2$-state (i.e., a state labeled by $select_2$)  whenever the first symbol of the reverse of a TM block $bl_2$ of $\beta_2$ is popped.
 The $select_2$-copy, which keeps track of $bit_i^{1}$, $u_1$, and $dir$, deterministically pops $bl_2^{R}$  and memorizes by its finite control
 the $i$-th bit  $bit_i^{2}$ of $bl_2$ and $\CON(bl_2)$. When $\CON(bl_2)=(u_{2,p},u_2,u_{2,s})$ is popped, then the $select_2$-copy terminates its computation, and moves to a $good_2$-state \emph{iff} $bit_i^{2}=bit_i^{1}$ and $u_1=next_{dir}(u_{2,p},u_2,u_{2,s})$.

 Let $ST$ be a strategy tree of $\mathcal{S}_2$ from  state $(p_0,\alpha^{R})$. For each $check_2$-node $x$ of $ST$, let $main(x)$ be the last $main_2$-node in the partial path from the root to $x$. 
  Let $x$ and $y$ be two distinct $check_2$-nodes of $ST$ which have the same distance from the root and such that $main(x)=main(y)$. First, we observe that the stack contents of $x$ and $y$ coincide, and $x$ and $y$ are associated with two distinct $check_2$-copies.  Since for all control states $p$,  $\vis(p)=\mu(p)$, it follows that for each $p\in \{check_2,select_2\}$, $x$ has a $p$-child iff $y$ has a $p$-child. Assume that $ST$ satisfies the fixed $\ACTL$ formula $\varphi_{check_2}$.
  Let $x$ be an arbitrary main node of $ST$ such that the stack content of $x$ is of the form $(bl_1^{R}\cdot\beta^{R}_1)\cdot dir\cdot \beta^{R}_2\cdot\alpha'$, where $bl_1$ is  a TM block, $\beta_1\cdot bl_1$ is the prefix of a TM code, $dir\in \{L,R\}$, and $\beta_2$ is a TM code.
 Let $\CON(bl_1)=(u_{1,p},u_1,u_{1,s})$. By construction, it follows that for each $1\leq i\leq n$, $x$ has a $check_2$-child $x_i$ such that  the subtree rooted at $x_i$ is a chain which leads to a TM $select_2$-block $bl_2^{i}$ of $\beta_2$ followed by a $good_2$-node such that the $i$-th bit of $bl_2^{i}$ coincides with the $i$-th bit of $bl_1$ and $u_1=next_{dir}(u_{2,p},u_2,u_{2,s})$, where $(u_{2,p},u_2,u_{2,s})=\CON(bl_2^{i})$. Moreover, by the observation above, it follows that all the $n$ $check_2$-copies associated with the $n$ $check_2$-children of $x$ select  the  same TM block $bl_2$  of $\beta_2$. Since
 the $i$-th bit of $bl_2$ coincides with the $i$-th bit of $bl_1$ for each $1\leq i\leq n$, $bl_2$ is precisely the TM block of $\beta_2$ have the same cell number as $bl_1$. It follows that $\alpha$ is faithful to the evolution of $\mathcal{T}$. Vice versa, if $\alpha$ is faithful to the evolution of $\mathcal{T}$, it easily follows that there is a strategy tree  from $(p_0,\alpha^{R})$ satisfying $\varphi_{check_2}$.
 \end{proof}

 Let $\mathcal{S}_0,\mathcal{S}_1$, and $\mathcal{S}_2$ be the \OPDS of Lemmata~\ref{lemma:STEp1}, \ref{lemma:STEp2}, and \ref{lemma:STEp3}, respectively. W.l.o.g.~we assume that the sets of visible and invisible control variables of these \OPDS are pairwise disjoint. Hence,
  their sets of control states  are pairwise disjoint as well. The \OPD $\mathcal{S}$ satisfying 
  Theorem~\ref{theorem:mainLowerBound} is obtained from $\mathcal{S}_0,\mathcal{S}_1$, and $\mathcal{S}_2$ as: (1) the set of control states is the union of the sets of control states of $\mathcal{S}_0,\mathcal{S}_1$, and $\mathcal{S}_2$, and the initial control state is the initial control state of $\mathcal{S}_0$, (2) the transition relation contains all the transitions of $\mathcal{S}_0,\mathcal{S}_1$, and $\mathcal{S}_2$ and, additionally,  two \emph{internal} transitions  from the special terminal control state $p_{fin}$ of $\mathcal{S}_0$ to the initial control states of $\mathcal{S}_1$ and $\mathcal{S}_2$, respectively, and (3) the  labeling function and the partitioning in environment and system states are obtained from those of $\mathcal{S}_0,\mathcal{S}_1$, and $\mathcal{S}_2$ in the obvious way. Let $\varphi_{check_1}$ and $\varphi_{check_2}$ be the fixed $\ACTL$ formulas of Lemmata~\ref{lemma:STEp2} and \ref{lemma:STEp3}, and let $\varphi_{finite}=\Aop\F(\Aop\X\, \neg\mathit{true})$ be the fixed $\ACTL$ formula asserting that a (finitely-branching) tree is finite.\footnote{note that a strategy tree of a \OPD is finitely-branching, i.e. the set of children of any node is finite.} Note that a  state of $\mathcal{S}$ is a state of $\mathcal{S}_0$ iff it is \emph{not} labeled by any proposition in $Prop_{fixed}=\{main_1,main_2,check_1,check_2,good_1,good_2,select_2\}$. \mbox{By  Lemmata~\ref{lemma:STEp1}, \ref{lemma:STEp2}, and \ref{lemma:STEp3}, we easily obtain that} \vspace{0.1cm}

 \noindent \textbf{Claim:} there is an accepting fair tree-code (i.e., $\mathcal{T}$ accepts $w_{in}$) iff there is a strategy tree of $\mathcal{S}$ satisfying the fixed $\ACTL$ formula $\varphi_{finite}\,\wedge\,\Aop\G\bigl([\bigwedge_{p\in Prop_{fixed}} \neg p]\, \longrightarrow\, [\bigwedge_{i=1}^{i=2}\Aop\X(main_i\rightarrow\varphi_{check_i})]\bigr)$.\vspace{0.1cm}

\noindent  By the claim above, Theorem~\ref{theorem:mainLowerBound} follows, which concludes.

\section{Pushdown module checking  for $\OPD$ with visible stack content depth}\label{sec:Problem2}

\dettagli{
Pushdown module checking $($with imperfect information$)$ for $\CTL$ specifications is known to be undecidable~\cite{AMV07}. However, the undecidability proof in~\cite{AMV07} uses \OPD for which each stack symbol consists entirely of invisible variables (hence, the environment does not see the push at all). The decidability status for the considered problem restricted to the class of \OPD where the stack content depth is visible is left open in~\cite{AMV07}. In this
section, we show that pushdown module checking for $\CTL$ specifications and \OPD with \emph{visible stack content depth} remains undecidable, even if the $\CTL$ formula is assumed to be in the fragment $\CTL(\E\F,\E\X,\Aop\G,\Aop\X)$ and the \OPD is assumed to be stable and having only environment configurations.
On the other hand, we show that pushdown module checking for the class of stable \OPD with visible stack content depth and having only environment configurations become decidable for  the existential fragment of $\CTL$.}

\subsection{Undecidability results}\label{sec:UndecidabilityProblem2}

In this subsection, we establish the following result.

\begin{theorem}\label{theorem:undecidabilityFirstStep} \PMC with imperfect information against $\CTL$ restricted to  \OPDS with \emph{visible stack content depth} is undecidable, even if the $\CTL$ formula is assumed to be in the fragment $\CTL(\E\F,\E\X,\Aop\G,\Aop\X)$ and the \OPD is assumed to be stable and having only environment configurations.
\end{theorem}
 Theorem~\ref{theorem:undecidabilityFirstStep} is proved  by a reduction from the Post's Correspondence Problem (PCP, for short) \cite{HU79}.
 An instance $\mathcal{I}$ of PCP is a tuple $\mathcal{I}=((u_1^{1},\ldots,u_n^{1}),(u_1^{2},\ldots,u_n^{2}))$, where $n\geq 1$ and
 for each $1\leq i\leq n$, $u^{1}_i$ and $u^{2}_i$ are non-empty finite words over an alphabet $A$. Let $[n]=\{1,\ldots,n\}$. A \emph{solution of}
 $\mathcal{I}$ is a non-empty sequence $i_1,i_2,\ldots,i_k$ of integers in $[n]$ such that
 $u^{1}_{i_1}\cdot u^{1}_{i_2}\cdot\ldots\cdot u^{1}_{i_k}=u^{2}_{i_1}\cdot u^{2}_{i_2}\cdot\ldots\cdot u^{2}_{i_k}$. PCP consists in checking for a given instance
 $\mathcal{I}$, whether $\mathcal{I}$ admits a solution. This problem is known to be undecidable \cite{HU79}. In the rest of this section, we fix a PCP instance
   $\mathcal{I}=((u_1^{1},\ldots,u_n^{1}),(u_1^{2},\ldots,u_n^{2}))$ and prove the following result, hence Theorem~\ref{theorem:undecidabilityFirstStep} follows.

\begin{theorem}\label{theorem:mainUndecidability} One can build  a \emph{stable} \OPD $\mathcal{S}$ with \emph{visible stack content depth} and having only environment configurations, and a $\CTL(\E\F,\E\X,\Aop\G,\Aop\X)$ formula $\varphi$ such that $\mathcal{I}$ has \emph{no} solution \emph{iff}  $\mathcal{M}_{\System}\models_r\varphi$.
\end{theorem}

In order to prove Theorem~\ref{theorem:mainUndecidability}, first we describe a suitable encoding of the set of solutions of $\mathcal{I}$.
Some ideas in the proposed encoding are taken from \cite{ACM06}, where  emptiness of alternating automata on nested trees is shown to be undecidable. \vspace{0.15cm}

\noindent \textbf{Preliminary step: encoding of the set of solutions of $\mathcal{I}$}. We use the following set $AP$ of atomic propositions: $AP = A\cup [n]\cup ([n]\times \{\natural\})\cup \{\flat,end_1,end_2,prev,succ,no_{match},match,\top_1,\top_2,\bot_1,\bot_2, \diamondsuit\}$.\vspace{0.1cm}

\noindent We denote by $MAX$  the maximum of the sizes of the words in $\mathcal{I}$ and by  $A^{MAX}$  the set of words $w\in A^{+}$ such that $|w|\leq MAX$.  Let $i_1,\ldots,i_k\in [n]^{+}$ (i.e., a non-empty sequence of integers in $[n]$) and $w\in A^{+}$ (i.e., a non-empty finite  word over $A$). A \emph{marked $(i_1,\ldots,i_k,w)$-word} is a finite word $v$  over $AP$ obtained from the word $\flat\cdot i_1\cdot \ldots\cdot i_k\cdot end_1\cdot w^{R}\cdot end_2$ by replacing at most  one integer occurrence $i_j$, where $1\leq j\leq k$, with $(i_j,\natural)$. The marked $(i_1,\ldots,i_k,w)$-word $v$ is \emph{good} if it contains exactly one marked integer occurrence.
A (good) \emph{marked word} is a (good) marked $(i_1,\ldots,i_k,w)$-word for some $i_1,\ldots,i_k\in [n]^{+}$ and  $w\in A^{+}$.  A \emph{marked tree} $T_{marked}$ is a \emph{minimal}   $AP$-labeled tree satisfying the following:
\begin{itemize}
  \item each finite path of $T_{marked}$ is labeled by a marked word;
  \item for all $i_1,\ldots,i_k\in [n]^{+}$ and $w\in A^{+}$, if there is a finite  path of $T_{marked}$ labeled by a marked $(i_1,\ldots,i_k,$ $w)$-word, then for each marked $(i_1,\ldots,i_k,w)$-word $v$, there is a path of $T_{marked}$ labeled by $v$.
  \item each infinite path of $T_{marked}$ is labeled by a word in $\{\flat\}\cdot [n]^{\omega}\cup \{\flat\}\cdot [n]^{*}\cdot [n]\times \{\natural\} \cdot [n]^{\omega}\cup \{\flat\}\cdot [n]^{*}\cdot [n]\times \{\natural\} \cdot [n]^{*}\cdot \{end_1\}\cdot A^{\omega}$.\footnote{this last condition is irrelevant in the encoding of the set of solutions of $\mathcal{I}$. It just reflects, as we will see, the behavior of the \OPD of Theorem~\ref{theorem:mainUndecidability}}.
\end{itemize}
Note that $i_1,\ldots,i_k$ is a solution of $\mathcal{I}$ iff there is a word $w\in A^{+}$ which can be factored into
$u^{1}_{i_1}\cdot u^{1}_{i_2}\cdot \ldots\cdot u^{1}_{i_k}$ and similarly into $u^{2}_{i_1}\cdot u^{2}_{i_2}\cdot \ldots\cdot u^{2}_{i_k}$. In order to express this condition, we define suitable extensions of the marked trees. First, we need additional definitions.

For each $t=1,2$, a $t$-witness for  $w$ is a \emph{finite minimal} $AP$-labeled tree $T^{t}_w$ satisfying the following:  $T^{t}_w$ consists of  a \emph{main path}  labeled by a word of the form  $\bot_t\cdot w_1\cdot \top_t\cdot \ldots\cdot \top_t\cdot  w_l\cdot \top_t$ such that:
 \begin{itemize}
   \item $w_1,\ldots,w_l\in A^{MAX}$ and $w_1\cdot \ldots\cdot  w_l=w$;
   \item each $\top_t$-node  has an additional child $x$, which does not belong to the main path, such that the subtree rooted at $x$ is a finite chain (called \emph{secondary chain}), whose nodes are labeled by $\diamondsuit$.
 \end{itemize}

\noindent  Let $x_i$ be the $i^{th}$  $\top_t$-node along the main path, where $1\leq i\leq l$: we denote by $\length(x_i)$  the length of the associated secondary chain,  by
$\word(x_i)$ the word $w_i$, and by $\suffix(x_i)$  the (possibly empty) word $w_{i+1},\ldots,w_l$. An \emph{extension} of a $t$-witness $T^{t}_w$ for  $w$ is a \emph{finite minimal} $AP$-labeled tree $ET^{t}_w$ obtained from $T^{t}_w$ by extending each secondary chain of $T^{t}_w$ with an additional (leaf) node labeled by a symbol in  $\{prev,succ,no_{match},match\}$. We say that $T^{t}_w$ is the \emph{support} of $ET^{t}_w$. For  $p\in \{prev,succ,no_{match},match\}$, we say that a $\top_t$-node of $ET^{t}_w$ is \emph{of type $p$} if the secondary chain associated with $x$ lead to a $p$-node.  Given a good  marked
$(i_1,\ldots,i_k,w)$-word  $v=\flat\cdot i_1\cdot \ldots\cdot i_{j-1}\cdot (i_j,\natural)\cdot \ldots\cdot i_k\cdot  end_1\cdot w^{R}\cdot end_2$, we say that $ET^{t}_w$ is compatible with $v$ iff for each
 $\top_t$-node $x$ along the main path of $ET^{t}_w$, the following holds:
\begin{itemize}
  \item  $\length(x)\in \{|\suffix(x)|+1,\ldots,|\suffix(x)|+k\}$. Moreover, if $\length(x)> |\suffix(x)|+k-j+1$ (resp., $\length(x)< |\suffix(x)|+k-j+1$), then $x$ is of type `$prev$' (resp., `$succ$');
  \item if  $\length(x)=|\suffix(x)|+k-j+1$ and $\word(x)= u_{i_j}^{t}$ (resp., $\word(x)\neq  u_{i_j}^{t}$), then
         $x$ is of type `$match$' (resp., `$no_{match}$').
\end{itemize}

A \emph{marked tree with witnesses} $WT_{marked}$ is a \emph{minimal} $AP$-labeled tree such that there is a marked tree $T_{marked}$ so that
 $WT_{marked}$ is obtained from $T_{marked}$ as follows:
 \begin{itemize}
   \item for each leaf $x$ of $T_{marked}$ (note that $x$ is  an $end_2$-node), let $v$ be the marked word labeling the partial path from the root to $x$. Then, if $v$ is good, we add two children $x_1$ and $x_2$ to $x$ such that for each $t=1,2$, the subtree rooted at $x_t$ is an extension of a $t$-witness compatible with 
       $v$;
   \item \emph{well-formedness requirement}: let $w\in A^{+}$ and $i_1,\ldots,i_k\in [n]^{+}$, and $x$ and $y$ be two $end_2$-nodes of
   $WT_{marked}$ such that the associated marked words 
   are good $(i_1,\ldots,i_k,w)$-marked words. Then, we require that for each $t=1,2$, the two subtrees rooted at the $\bot_t$-child of $x$ and $y$, respectively, (which are extensions of $t$-witnesses) have the same support.
 \end{itemize}

\begin{proposition}\label{remark:PCP}$\mathcal{I}$ admits a solution \emph{iff} there is a \emph{marked tree with witnesses} $WT_{marked}$ having some $end_2$-node and such that 
for each  $\bot_t$-node $x$ ($t=1,2$), the subtree $ET_{w}^{x}$ rooted at  $x$ satisfies the following:
\begin{itemize}
  \item $ET_{w}^{x}$ has no `$no_{match}$'-nodes and there is exactly one node of $ET_{w}^{x}$ which is labeled by `$match$';
  \item \emph{no} $\top_t$-node   of type `${match}$' or `${succ}$'  is strictly followed by a
  $\top_t$-node  of type `${match}$' or `${prev}$'.
\end{itemize}
\end{proposition}

\noindent By Proposition~\ref{remark:PCP}, 
we easily  deduce the following. 

\begin{proposition}\label{prop:PCP}
One can construct a $\CTL(\E\F,\E\X,\Aop\G,\Aop\X)$ formula $\psi_{\mathcal{I}}$ such that $\mathcal{I}$ admits a solution \emph{if and only if} there is a \emph{marked tree with witnesses} $WT_{marked}$  which satisfies $\psi_{\mathcal{I}}$.
\end{proposition}

Since $\CTL(\E\F,\E\X,\Aop\G,\Aop\X)$ is closed under negation, Theorem~\ref{theorem:mainUndecidability} directly follows from Proposition~\ref{prop:PCP} and the following lemma.

\begin{lemma}\label{lemma:MainLemmaUndecidability} One can construct a \emph{stable} \OPD $\mathcal{S}$ with \emph{visible stack content depth} and having only environment configurations, and a $\CTL(\E\F,\E\X,\Aop\G,\Aop\X)$ formula $\phi$ such that the set of strategy trees of $\mathcal{S}$ which satisfy $\phi$ corresponds to the set of marked trees with witnesses.
\end{lemma}
\begin{proof}
We informally describe the construction of the stable \OPD $\mathcal{S}=\tpl{AP,Q,q_0,\Gamma,\flat,\Delta,\mu,Env}$.
 Each state of $\mathcal{S}$ is an environment state, i.e. $Env=Q\times (\Gamma\cup \{\flat\})$, and  the  labeling function $\mu$  can be seen as mapping $\mu:Q \rightarrow AP$. 
The sets $I_\Gamma$ and $H_\Gamma$
of visible and invisible stack content variables are given by $I_\Gamma =A \cup [n]$ and $H_\Gamma=\{\natural\}$. Then, $\Gamma$ is given by $\Gamma=\{\{\gamma\}\mid \gamma\in I_\Gamma\}\cup \{\{i,\natural\}\mid i\in [n]\}$. We identify $\{\gamma\}$ with $\gamma$ and  $\{i,\natural\}$ with $(i,\natural)$. Hence, $\Gamma$ corresponds to the set
      $A\cup [n]\cup ([n]\times \{\natural\})$. Note that $\vis(\gamma)\neq\varepsilon$ for each $\gamma\in\Gamma$. Hence, the stack content depth of $\mathcal{S}$ is visible and:

\begin{itemize}
  \item \textbf{Property A:} for all $\gamma,\gamma'\in\Gamma$, $\vis(\gamma)=\vis(\gamma')$ iff either $\gamma=\gamma'$ or
      $\gamma,\gamma'\in \{i,(i,\natural)\}$  for some $i\in [n]$.
\end{itemize}

\noindent Furthermore, the definition of   $\mu$ and $P$ ensures the following: 

\begin{itemize}
  \item \textbf{Property B:} for all $q,q'\in P$, $\vis(q)=\vis(q')$ iff: (1) or $\mu(q)=\mu(q')$, or (2)
      $\mu(q),\mu(q')\in\{i,(i,\natural)\}$ for some $i\in [n]$, or (3)
      $\mu(q),\mu(q')\in \{no_{match},match,prev,succ\}$.\footnote{In fact, in order to ensure that $\mathcal{S}$ is stable, Property B is slightly more complicated. 
      }
\end{itemize}


\noindent \emph{First phase: generation of marked words.} Starting from the initial configuration (whose stack content and propositional label is $\flat$), the \OPD $\mathcal{S}$ generates  symbol by symbol,\footnote{i.e., the transitions in this phase lead to configurations labeled by propositions in $\{end_1,end_2\}\cup A\cup [n]\cup ([n]\times \{\natural\})$} by external nondeterminism, marked words. Whenever a symbol in
$A\cup [n]\cup ([n]\times \{\natural\})$ is generated, at the same time it is pushed onto the stack. Symbols in $\{end_1,end_2\}$ are generated by internal transitions that do not modify the stack content. The \OPD $\mathcal{S}$ keeps track by its finite control whether there is a marked integer in the prefix of the guessed marked word generated so far. In such a way, $\mathcal{S}$ can ensure that during the generation of a marked word, at most one integer occurrence in $[n]$ is marked. Let
 $\Upsilon$  be the set of $AP$-labeled trees $T$  such that there is a strategy tree $ST$ of $\mathcal{S}$  so that $T$ is obtained from $ST$ by pruning the subtrees rooted at the children of $end_2$-nodes. Then, Properties~A and B above ensure that $\Upsilon$ is  the \emph{set of marked trees}. \vspace{0.1cm}

\noindent \emph{Second phase: generation of extensions of $t$-witnesses, where $t=1,2$.}
Assume that $\mathcal{S}$ is in an $end_2$-state $s$ associated with some node $x_s$ of the  computation tree of $\mathcal{S}$ from the initial state. By construction, the partial path from the root to $x_s$ is labeled by some marked word $v$. If $v$ is not good, then $s$ has no successors. Now, assume that $v$ is good, hence, $v$ is of the form $\flat\cdot i_1\cdot\ldots\cdot(i_j,\natural)\cdot\ldots\cdot i_k\cdot end_1\cdot w^{R}\cdot end_2$, where $w\in A^{+}$ and $i_1,\ldots,i_k\in [n]^{+}$. By construction, the stack content in $s$ is given by $w\cdot i_k\cdot\ldots\cdot(i_j,\natural)\cdot\ldots\cdot i_1\cdot\flat$.
Then, from state $s$, $\mathcal{S}$ splits in two copies: the first one moves to a configuration $s_1$ labeled by $\bot_1$ and the second one moves to configuration $s_2$ labeled by $\bot_2$ (in both cases the stack content is not modified). Fix $t=1,2$. From state $s_t$, $\mathcal{S}$ generates by external nondeterminism extensions of $t$-witnesses compatible with the marked word $v$ as follows. Finite words of the form
$w_1\cdot\top_t\cdot\ldots\cdot\top_t\cdot w_l\cdot\top_t$, where $w_1,\ldots,w_l\in A^{MAX}$ and $w_1\cdot\ldots\cdot w_l=w$, labeling main paths of $t$-witnesses, are generated as follows. The symbol $\top_t$ is generated by internal  transitions which do not modify the stack content. Whenever the symbol $\bot_t$ (resp., $\top_t$) is generated,  $\mathcal{S}$ pops (resp., can pop) the stack symbol by  symbol  and generates the current popped symbol (with the restriction that a symbol can be popped iff it is in $A$). At the same time, $\mathcal{S}$ keeps track by its finite control of the string $w_s\in A^{MAX}$ popped so far. When $|w_s|=MAX$, then $\mathcal{S}$ deterministically moves to a $\top_t$-configuration (without changing the stack content).
If instead $|w_s|<MAX$, then $\mathcal{S}$ either continues to pop the stack content (if the top of the stack content is in $A$) or moves to a $\top_t$-configuration (without changing the stack content).  Additionally, from a $\top_t$-configuration, $\mathcal{S}$ can also choose to move to a $\diamondsuit$-configuration $s_{\diamondsuit}$ without changing the stack content. In $s_{\diamondsuit}$, $\mathcal{S}$
 keeps track in the control state of the word $w_s\in A^{MAX}$ (popped from the stack) and associated with the previous $\top_t$-configuration. Starting from $s_{\diamondsuit}$, $\mathcal{S}$ deterministically pops the stack symbol by symbol remaining in  $s_{\diamondsuit}$. When every symbol in $A$ has been popped (hence, the stack content is $i_k\cdot\ldots\cdot(i_j,\natural)\cdot\ldots\cdot i_1\cdot\flat$), $\mathcal{S}$ can choose to continue to pop the stack symbol by symbol by moving at each step to $\diamondsuit$-configurations and by keeping track in its finite control of the string $w_s$ and whether a marked integer in $[n]$ has been already popped. Additionally, whenever a symbol in $[n]\cup [n]\times \{\natural\}$ is popped,  $\mathcal{S}$ can choose to move without changing the stack content to a terminal $p$-configuration, where $p\in\{prev,succ,match,no_{match}\}$, such that the following holds: $p=succ$ (resp., $p=prev$) if an integer in $[n]$ is popped and no (resp., some) marked integer has been previously popped, and $p=match$ (resp., $p=no_{match}$) if a marked integer $(h,\natural)$ (note that $h=i_j$) is popped and $w_s=u^{t}_h$ (resp., $w_s\neq u^{t}_h$).

 We use the following $\CTL(\E\F,\E\X,\Aop\G,\Aop\X)$ formula $\phi$ in order to select strategy trees of $\mathcal{S}$ such that: (1) each $end_2$-node has two children (i.e., a child labeled by $\bot_1$ and a child labeled $\bot_2$), and (2) for each $t=1,2$, the subtree rooted at any $\bot_t$-node is an extension of a $t$-witness.
 In order to fulfill the second requirement, first, we need to ensure that from each $\bot_t$ node ($t=1,2$), there is a unique main path. Note that this last condition is equivalent to require that each $a$-node with $a\in A$ in a $\bot_t$-node rooted subtree has exactly one child (this can be easily expressed in $\CTL(\E\F,\E\X,\Aop\G,\Aop\X)$, since the strategies trees of $\mathcal{S}$ are \emph{minimal} $AP$-labeled trees). Second, we need to ensure that each $\top_t$-node has a $\diamondsuit$-child $x$ such that the subtree rooted at $x$ is a finite chain. 
 Hence, 
 formula $\phi$ is given by
 \vspace{0.1cm}

 \noindent \text{\hspace{0.0cm}} $\Aop\G(end_2 \rightarrow \bigwedge_{t=1,2}\E\X (\bot_t\,\wedge\,\Aop\G[(\bigvee_{a\in A}a \rightarrow \psi_{unique})
         \,\wedge\,(\top_t\rightarrow\E\X \diamondsuit)\,\wedge\,(\diamondsuit \rightarrow (\psi_{unique}\wedge \E\F\,\Aop\X\neg \mathtt{true}))]))$\vspace{0.1cm}

\noindent where $\psi_{unique}=\bigvee_{p\in AP}\Aop\X p$. By Properties A and B above it easily follows that the strategy trees of $\mathcal{S}$ satisfying the $\CTL(\E\F,\E\X,\Aop\G,\Aop\X)$ formula $\phi$, also satisfy the well-formedness requirement. Hence, the set of strategy trees of $\mathcal{S}$ satisfying $\phi$ is the set of marked  trees with witnesses.
\end{proof}

\subsection{Decidability results}

The main result of this subsection is as follows.

\begin{theorem}\label{theorem:DecidabilityMAIN} \PMC with imperfect information against $\ECTL$ restricted to stable \OPDS with visible stack content depth and having only environment configurations is decidable and in \TWOEXPTIME.
\end{theorem}

 Theorem~\ref{theorem:DecidabilityMAIN} is proved by a reduction to non-emptiness of B\"{u}chi alternating visible pushdown automata (\AVPA) \cite{Boz07}, which is  \TWOEXPTIME-complete \cite{Boz07}. First, we briefly recall the framework of \AVPA. Then, we establish some additional decidability results. Finally, we prove Theorem~\ref{theorem:DecidabilityMAIN}.\vspace{0.2cm}

\noindent \textbf{B\"{u}chi \AVPA:}  A
\emph{pushdown alphabet} $\Sigma$ is a finite alphabet which is
partitioned in three disjoint finite alphabets $\Sigma^{call}$,
$\Sigma^{ret}$, and $\Sigma^{int}$, where $\Sigma^{call}$ is a  set
of \emph{calls}, $\Sigma^{ret}$ is a  set of \emph{returns}, and
$\Sigma^{int}$ is a  set of \emph{internal actions}.
An \AVPA is a standard alternating pushdown automaton on  words over a pushdown alphabet $\Sigma$, which  pushes onto (resp., pops) the stack only when
it reads a call (resp., a return), and does not use the stack on
internal actions. For a formal definition of the syntax and semantics of $\AVPA$ see \cite{Boz07}. 
Given a
B\"{u}chi \AVPA $\mathcal{A}$ over $\Sigma$, we denote by $\mathcal{L}(\mathcal{A})$ the set of nonempty finite or infinite words over $\Sigma$ accepted by $\mathcal{A}$ (we assume that $\mathcal{A}$ is equipped with both a B\"{u}chi acceptance condition for infinite words and a standard acceptance condition for finite words).\vspace{0.2cm}

\noindent \textbf{Preliminary decidability results:} For a module $\mathcal{M}$, a \emph{minimal} strategy tree $ST_{min}$ of $\mathcal{M}$ is a strategy tree satisfying the following: for each strategy tree $ST$ of $\mathcal{M}$ if $ST$ is contained in $ST_{min}$, then $ST=ST_{min}$.
Given a $\CTL$ formula $\varphi$, we say that  $\mathcal{M}$ \emph{minimally reactively satisfies} $\varphi$, denoted $\mathcal{M}\models_{r,min}\varphi$, if all the \emph{minimal} strategy trees of $\mathcal{M}$  satisfy $\varphi$.
Let $\mathcal{M}$ be a \emph{stable} module having only environment states and $ST$ be a minimal strategy tree of $\mathcal{M}$. For each $i\geq 0$, let $\Lambda_i$ be the set of nodes $x$ of $ST$ at distance $i$ from the root, i.e.,  such that $|x|=i$. Since $ST$ is minimal, it easily follows that for all $i\geq 0$ and $x,x'\in\Lambda_i$, $\vis(\state(x))=\vis(\state(x'))$. Now, let us consider a stable \OPD $\mathcal{S}=\tpl{AP,Q,q_0,\Gamma,\flat,\Delta,\mu,Env}$ with visible stack content depth and having only environment configurations.  By Remark~\ref{remark:stable}, $\mathcal{M}_{S}$ is stable. Let $ST$ be a minimal strategy tree of $\mathcal{S}$ and for each $i\geq 0$, let $\Lambda_i$ be defined as above (w.r.t.~strategy $ST$).  By the above observation, it easily follows that for each $i\geq 0$ such that $\Lambda_{i+1}\neq\emptyset$, there are $X_i\subseteq I$ (where $I$ is the set of visible control state variables of $\mathcal{S}$) and
$X_{i,\Gamma}\subseteq I_\Gamma$ (where $I_\Gamma$ is the set of visible stack content variables of $\mathcal{S}$) such that one of the following holds:
\begin{itemize}
  \item each node $x$ in $\Lambda_{i+1}$ is obtained from the parent node by an internal transition (depending on $x$) of the form
      $q\der{}q'$ such that $\vis(q')=X_i$;
  \item each node $x$ in $\Lambda_{i+1}$ is obtained from the parent node by a push transition (depending on $x$) of the form
      $q\der{\mathtt{push}(\gamma)}q'$ such that $\vis(q')=X_i$ and $\vis(\gamma)=X_{i,\Gamma}$;
  \item each node $x$ in $\Lambda_{i+1}$ is obtained from the parent node by a pop transition (depending on $x$) of the form
      $q\der{\mathtt{pop}(\gamma)}q'$ such that $\vis(q')=X_i$.
\end{itemize}

Let $\Sigma_{\mathcal{S}}$ be the pushdown alphabet defined as follows:
$\Sigma_{\mathcal{S}}^{call}=\{(push,X,X_\Gamma)\mid X=\vis(q)$ and $X_\Gamma=\vis(\gamma)$ for some $q\in Q$ and $\gamma\in\Gamma\}$,
$\Sigma_{\mathcal{S}}^{int}=\{(int,X)\mid X=\vis(q)$ for some $q\in Q\}$, and
$\Sigma_{\mathcal{S}}^{ret}=\{(pop,X)\mid X=\vis(q)$ for some $q\in Q\}$. Thus, we can associate to each finite (resp., infinite) minimal strategy tree $ST$ of
$\mathcal{S}$ a finite (resp., infinite) word over $\Sigma_{\mathcal{S}}$, denoted by $w(ST)$. Moreover, for each word $w$ over $\Sigma_{\mathcal{S}}$, there is at most one minimal strategy tree $ST$ of $\mathcal{S}$ such that $w(ST)=w$. This  observation leads to the following theorem, where $\widehat{\Sigma}_{\mathcal{S}}$
is the pushdown alphabet $\Sigma_{\mathcal{S}}\cup\{push,pop\}$, with $push$ being a call, and $pop$ a return.

\begin{theorem}\label{theoremFromOPDtoAVPA} Given a stable \OPD $\mathcal{S}$ with visible stack content depth and having only environment configurations and a $\CTL$ formula $\varphi$, one can construct in linear-time a  B\"{u}chi \AVPA $\mathcal{A}$ over $\widehat{\Sigma}_{\mathcal{S}}$  such that there is a \emph{minimal} strategy tree of $\mathcal{S}$ satisfying $\varphi$ iff $\mathcal{L}(\mathcal{A})\neq\emptyset$.
\end{theorem}
\begin{proof}
The proposed construction is a generalization of the standard alternating automata-theoretic approach to $\CTL$ model checking~\cite{KVW00}.
Here, we informally describe the main aspects of the construction. Let $\mathcal{S}=\tpl{AP,P,p_o,\Gamma,\flat,\Delta,\mu,Env}$.
W.l.o.g.~we assume that the initial configuration of $\mathcal{S}$ is non-terminal.
For a word $w$ over $\Sigma_\mathcal{S}$, we denote by $ext(w)$ the  word over $\widehat{\Sigma}_{\mathcal{S}}$ obtained from $w$ by replacing each occurrence of a return symbol $(pop,X)$ in $w$ with the word
$(pop,X),pop,push$. We construct a  B\"{u}chi \AVPA $\mathcal{A}$ over $\widehat{\Sigma}_{\mathcal{S}}$ such that for each non-empty  word $\widehat{w}$ over $\widehat{\Sigma}_{\mathcal{S}}$, $\mathcal{A}$ has an accepting run over $\widehat{w}$ \emph{if and only if} $\widehat{w}=ext(w)$ for some  word $w$ over $\Sigma_\mathcal{S}$ and there is a minimal strategy tree $ST$ of $\mathcal{S}$ such that $w=w(ST)$ and $ST$  satisfies $\varphi$.
Essentially, for each  word $w$ over $\Sigma_\mathcal{S}$ associated with some minimal strategy tree $ST$ of $\mathcal{S}$, an accepting run $r$ of $\mathcal{A}$ over $ext(w)$ encodes $ST$ as follows: the nodes of $r$ associated with the $i$-th symbol of $w$ correspond to the nodes  of $ST$ at distance $i$ from the root. However, for each node $x$ of $ST$, there can be many copies of $x$ in the run $r$. Each of such copies has the same stack content as $x$, but its control state  is equipped with additional information including one of the subformulas of $\varphi$ which holds at node $x$ of $ST$.

The \AVPA $\mathcal{A}$ has the same stack alphabet as $\mathcal{S}$. Its set of control states is instead given by the set of tuples of the form $(p,\gamma,\psi,f)$, where $(p,\gamma)\in P\times (\Gamma\cup\{\flat\})$, $\psi$ is a subformula of $\varphi$, and $f$ is an additional state variable in $\{sim,pop,push\}$. Intuitively, $p$ represents the current control state of $\mathcal{S}$ and $\gamma$ represents the guessed top symbol of the current stack content. Furthermore,  $f$ is used to check that the input word is an extension of some  word over $\Sigma_\mathcal{S}$.
The additional symbols $pop$ and $push$ in  $\widehat{\Sigma}_\mathcal{S}$ are instead used to check that the guess $\gamma$ is correct. The behavior of $\mathcal{A}$ as follows. Assume that a copy of $\mathcal{A}$ is in a control state of the form $(p',\gamma',\psi',sim)$ and the current input symbol is $\sigma$, where
  $p'$ is the current control state of $\mathcal{S}$ and $\gamma'$ is the top symbol of the current stack content (initially, $\mathcal{A}$ is in the control state $(p_0,\flat,\varphi,sim)$).  If $\sigma\in\{pop,push\}$, then the input is rejected. If instead $\sigma$ is call (resp., an internal action) in $\Sigma_{\mathcal{S}}$, then the considered copy of $\mathcal{A}$ simulate push (resp., internal) transitions of $\mathcal{S}$ from the current configuration (of the form $(p',\alpha)$ such that $\TOP(\alpha)=\gamma'$) consistent with $\sigma$ if such  transitions exist by splitting in one or more copies (depending on the number of simulated transitions and the structure of $\psi$), each of them moving to a control state of the form
  $(p,\gamma,\psi,sim)$. Note that in this case, $\mathcal{A}$ can ensure that the guess $\gamma$ is correct. Now, assume that $\sigma$ is a return in
  $\Sigma_{\mathcal{S}}$. Then, the considered copy of $\mathcal{A}$ guesses a stack symbol $\gamma\in \Gamma\cup\{\flat\}$ and  simulate pop transitions
of $\mathcal{S}$ from the current configuration consistent with $\sigma$ (if such transitions exist) by splitting in  one or more copies (depending on the number of simulated transitions and the structure of $\psi$), each of them moving to a control state of the form
$(p,\gamma,\psi,pop)$. In the next step, the input symbol must be $pop$ (otherwise, the input is rejected). Thus, the current copy in control state $(p,\gamma,\psi,pop)$  pops the stack and check whether the guess $\gamma$ is correct. If the guess  is correct, then the copy moves to the control state  $(p,\gamma,\psi,push)$ (otherwise, the run is rejecting). In the next step, the input symbol must be $push$ (otherwise, the input is rejected). Thus, the considered copy re-pushes $\gamma$ onto the stack and moves to control state $(p,\gamma,\psi,sim)$.
  Assuming that the input word is $ext(w)$ for some nonempty word $w$ over $\Sigma_\mathcal{S}$, the above behavior
 ensures, in particular, that whenever
an input symbol in $\Sigma_{\mathcal{S}}$ is read, $\mathcal{A}$ is in a control state of the form $(p,\gamma,\psi,sim)$, where $\gamma$ is the top symbol of the current stack content.
Finally, $\mathcal{A}$ checks whether $w$ is associated with some  minimal strategy tree of $\mathcal{S}$ as follows.
First, we observe that a nonempty word $w$ over $\Sigma_{\mathcal{S}}$ is not associable to any minimal strategy tree of $\mathcal{S}$ iff the following holds. There is a proper prefix $w'$ of $w$ of length $i$ for some $i\geq 0$ such that $w'$ is the prefix of $w(ST)$ for some minimal strategy tree $ST$ of $\mathcal{S}$ such that: there is a node $x$ of $ST$ at distance $i+1$ from the root whose configuration $(p,\alpha)$ has some successor, but there is no transition from $(p,\alpha)$ which is consistent with the $i+1$-th symbol of $w$. Thus, whenever a copy of $\mathcal{A}$ reads a symbol $\sigma\in\Sigma_{\mathcal{S}}$, hence the considered copy is in a control state of the form $(p,\gamma,\psi,sim)$ (where $p$ is the current control state of $\mathcal{S}$ and $\gamma$ is the top symbol of the current stack content), $\mathcal{A}$ rejects the input string if: the current configuration of $\mathcal{S}$ has some successor (i.e., $(p,\gamma)$ is non-terminal), but there is no transition from the current configuration which is consistent with the current input symbol $\sigma$.
\end{proof}

Since non-emptiness of \AVPA is  \TWOEXPTIME-complete \cite{Boz07}, by Theorem~\ref{theoremFromOPDtoAVPA}, we obtain the following.

\begin{corollary}\label{cor:FromOPDtoAVPA} Checking whether $\mathcal{M}_{\mathcal{S}}\models_{r,min}\varphi$, for a given $\CTL$ formula $\varphi$ and a given  stable \OPD $\mathcal{S}$ with visible stack content depth and having only environment configurations,  is in \TWOEXPTIME.
\end{corollary}

\noindent \textbf{Proof of Theorem~\ref{theorem:DecidabilityMAIN}}: let $\varphi$ be an $\ECTL$ formula over $AP$. Note that for all $2^{AP}$-labeled trees $T$ and $T'$, if $T$ is contained in $T'$ and $T$ satisfies $\varphi$, then $T'$ satisfies $\varphi$ as well. Note that for a given module $\mathcal{M}$, each strategy tree of $\mathcal{M}$ contains some minimal strategy tree. Hence, for an $\ECTL$ formula $\varphi$, $\mathcal{M}\models_{r}\varphi$ if and only if $\mathcal{M}\models_{r,min}\varphi$. Thus,  Theorem~\ref{theorem:DecidabilityMAIN} directly follows from Corollary~\ref{cor:FromOPDtoAVPA}.
Finally, for completeness, we observe that unrestricted \PMC with imperfect information against $\ACTL$ is trivially decidable. Indeed
for an $\ACTL$ formula $\varphi$ and module $\mathcal{M}$,  $\mathcal{M}\models_{r}\varphi$ iff the \emph{maximal} strategy tree of $\mathcal{M}$ (i.e., the computation tree of $\mathcal{M}$ starting from the initial state) satisfies $\varphi$. Hence, \PMC with imperfect information against $\ACTL$  is equivalent to standard pushdown model checking against $\ACTL$, which is in \EXPTIME\ \cite{Wal00}.

\begin{proposition} \PMC with imperfect information against $\ACTL$ is in \EXPTIME.
\end{proposition}

\section{Conclusion}

There is an intriguing question left open. We have shown the \PMC with imperfect information for stable \OPDS with visible stack content depth and having only environment configurations is undecidable for the fragment $\CTL(\E\F,\E\X,\Aop\G,\Aop\X)$ of $\CTL$, and decidable  for the fragments $\ECTL$ and $\ACTL$ of $\CTL$. Thus, it is open the decidability status of the problem above for the standard $\E\F$-fragment of $\CTL$
(using just the temporal modality $\E\F$ and its dual $\Aop\G$). We conjecture that the problem is decidable.

\bibliographystyle{eptcs}
\bibliography{refGandalf}
\end{document}